\documentclass[
 reprint,
 amsmath,amssymb,
 aps,
longbibliography
]{revtex4-1}

\usepackage{graphicx}% Include figure files
\usepackage{dcolumn}% Align table columns on decimal point
\usepackage{bm}% bold math
\usepackage{color}
\usepackage{amsthm}
\usepackage{tikz}
\usepackage{amsfonts}
\usepackage{amssymb}
\usepackage[mathscr]{eucal}
\usepackage[normalem]{ulem}
\usepackage[colorlinks]{hyperref} %added this to have hyperlinks
\usepackage{braket}
\usepackage[only,llbracket,rrbracket]{stmaryrd} %to import only llbracket and rrbracket from the package and aboid conflicts

\usepackage{cleveref}

\newtheorem{lemma}{Lemma}
\newtheorem{thm}{Theorem}

\newcommand{\N}{{\sf N}}
\newcommand{\D}{\sf{D}}

\newcommand{\I}{\mathscr{I}}

\newcommand{\ent}{{\sf S}}  % 
\newcommand{\er}{\mathscr{R}} % extreme ray of SAC
\newcommand{\erq}[1]{\er_{#1}} % extreme ray of SAC consistent w/ QM
\newcommand{\face}{\mathscr{F}}   % Face of SAC (associated to PMI)

\newcommand{\qec}[1]{\overline{\Sigma}^*_{#1}} % symbol for the QEC
\newcommand{\poly}[1]{\Sigma_{#1}} % symbol for polyhedral cone obtained by entropy inequalities 

\definecolor{THcolor}{rgb}{0.1,0.7,0.1}

\definecolor{VHcolor}{rgb}{0.7,0.3,0.9}

\definecolor{MRocolor}{rgb}{0.1,0,1}

\begin{document}

\preprint{}

\hfill CALT-TH 2023-050

\title{Inner bounding the quantum entropy cone with subadditivity and subsystem coarse grainings}

\author{Temple He}
\affiliation{Walter Burke Institute for Theoretical Physics \\ California Institute of Technology, Pasadena, CA 91125 USA}
\email{templehe@caltech.edu}

\author{Veronika E. Hubeny}
\affiliation{Center for Quantum Mathematics and Physics (QMAP)\\ 
Department of Physics \& Astronomy, University of California, Davis, CA 95616 USA}
\email{veronika@physics.ucdavis.edu}

\author{Massimiliano Rota}
\affiliation{Center for Quantum Mathematics and Physics (QMAP)\\ 
Department of Physics \& Astronomy, University of California, Davis, CA 95616 USA}
\affiliation{Institute for Theoretical Physics, University of Amsterdam,
Science Park 904, Postbus 94485, 1090 GL Amsterdam, The Netherlands}
\email{maxrota@ucdavis.edu}

\date{\today}% It is always \today, today,

\begin{abstract}
We show via explicit construction that all the extreme rays of both the three-party quantum entropy cone and the four-party stabilizer entropy cone can be obtained from subsystem coarse grainings of specific higher-party quantum states, namely extreme states characterized by saturating a (non-trivial) maximal set of instances of subadditivity. 
This suggests that the study of the ``subadditivity cone'', and the set of its extreme rays realizable in quantum mechanics, provides a powerful tool for deriving inner bounds for the quantum and stabilizer entropy cones, as well as constraints on new inequalities for the von Neumann entropy.  
\end{abstract}

\maketitle

%\tableofcontents

%~~~~~~~~~~~~~~~~~~~~~~~~~~~~~~~~~~~~~~~~~~~~~~~~~~~~~~~~~~~~~~
\section{\label{sec:intro} Introduction}
%~~~~~~~~~~~~~~~~~~~~~~~~~~~~~~~~~~~~~~~~~~~~~~~~~~~~~~~~~~~~~~

The derivation of the fundamental inequalities satisfied by the von Neumann entropy is an important problem in quantum information theory. Unfortunately, it is also notoriously difficult. Since the proof of strong subadditivity (SSA) by Lieb and Ruskai half a century ago  \cite{1973ssa:Lieb,Lieb:1973cp}, no new \textit{unconstrained} inequality has been found. New inequalities involving four or more parties were discovered in \cite{2005:Linden,2011:Cadney,Christandl_2023}, but they are \textit{constrained}, meaning that they only apply to particular density matrices which saturate other entropic constraints. In the classical case, i.e., for the Shannon entropy, new inequalities have been derived in \cite{1998:Zhang, 2002:Makarychev, 2003:Zhang, 2007:Matus}, and \cite{2011:Cadney} speculated that they might hold also in quantum mechanics. However, to our knowledge, at present there is no strong argument that new unconstrained inequalities for the von Neumann entropy should exist. 

The main goal of this paper is to show that a seemingly much simpler problem, the \textit{quantum marginal independence problem} (QMIP) \cite{Hernandez-Cuenca:2019jpv}, originally motivated by an approach \cite{Hubeny:2018trv,Hubeny:2018ijt,Hernandez-Cuenca:2022pst} to the study of entropy inequalities in quantum gravity \cite{Ryu:2006bv,Hayden:2011ag, Bao:2015bfa}, might provide a powerful tool for the derivation of constraints on the existence of new inequalities and their structure. Specifically, we will show that already from a partial solution to the ``extremal'' version of the QMIP up to nine parties, we can derive definite \textit{inner bounds} for the quantum entropy cone, which in the case of three parties coincide with the full cone \cite{1193790}, and in the case of four parties coincide with the entropy cone of stabilizer states \cite{Linden:2013kal}. We stress that even though we are not in this paper deriving new bounds on quantum entropy cones for four or more parties, our method introduces a conceptual idea that may prove fruitful to explore further.

The structure of the paper is as follows. In \S\ref{sec:cones} and \S\ref{sec:qmip}, we review the main definitions concerning entropy cones, the QMIP, and its extremal version. In \S\ref{sec:inner_bounds}, we formulate the inner bound to the quantum entropy cone for an arbitrary number of parties. In \S\ref{sec:hypergraphs}, we review the definition of hypergraph models from \cite{Bao:2020zgx} and the main result from \cite{Walter:2020zvt} about their realizability by stabilizer states. We then use this technology in \S\ref{sec:examples} to derive the main result of this paper, and conclude in \S\ref{sec:discussion} with a list of open questions.

%~~~~~~~~~~~~~~~~~~~~~~~~~~~~~~~~~~~~~~~~~~~~~~~~~~~~~~~~~~~~~~
\section{The quantum entropy cone} 
\label{sec:cones}
%~~~~~~~~~~~~~~~~~~~~~~~~~~~~~~~~~~~~~~~~~~~~~~~~~~~~~~~~~~~~~~

A convenient framework to study entropy inequalities (and their relations) was introduced in \cite{1193790}, following the analogous work of \cite{641556} for the Shannon entropy. Let $[\N]=\{1,2,\ldots,\N\}$, and consider an $\N$-party density matrix $\rho_{\N}$ on a Hilbert space $\mathcal{H}_1\otimes\mathcal{H}_2\otimes\ldots\otimes\mathcal{H}_{\N}$. The \textit{entropy vector} of $\rho_{\N}$ is the vector in $\mathbb{R}^{\D}$, with $\D=2^{\N}-1$, given by
\begin{equation}
    \vec{\ent}=\{\ent_{\I},\;\forall\, \I\subseteq [\N]\},\qquad \ent_{\I}:=\ent(\rho_{\I}),
\end{equation}
where $\rho_{\I}$ is the reduced density matrix for the parties in $\I$. Denoting by $\Sigma_{\N}^*$ the set of entropy vectors for all $\N$-party quantum states, its topological closure $\qec{\N}$ is a convex cone \cite{1193790} known as the \textit{quantum entropy cone} (QEC). 

For any $\N$, an \textit{outer bound} of $\qec{\N}$ is given by the \textit{polyhedral} cone obtained via the intersection of the half-spaces specified by all instances of subadditivity (SA) and strong subadditivity (SSA) of the von Neumann entropy, and we will denote this cone by $\poly{\N}$ \cite{1193790}. For $\N\leq 3$ it is easy to show that $\qec{\N}=\poly{\N}$ by constructing quantum states that realize entropy vectors belonging to the extreme rays of $\poly{\N}$ \cite{1193790}. On the other hand, for $\N\geq 4$, the constrained inequalities found in \cite{2005:Linden,2011:Cadney} imply that this construction is no longer possible, since there are regions on the boundary of $\poly{\N}$ which are inaccessible to quantum states. This however is not enough to conclude that $\qec{\N}\subset\poly{\N}$ (strictly), since there might be quantum states that approximate the entropy vectors in these regions arbitrarily well.

%~~~~~~~~~~~~~~~~~~~~~~~~~~~~~~~~~~~~~~~~~~~~~~~~~~~~~~~~~~~~~~
\section{Quantum marginal independence} 
\label{sec:qmip}
%~~~~~~~~~~~~~~~~~~~~~~~~~~~~~~~~~~~~~~~~~~~~~~~~~~~~~~~~~~~~~~

An interesting and well-known problem in quantum mechanics is the \textit{quantum marginal problem} \cite{Schilling:2014qmp}. Given density matrices $\rho_{\I}$ for some subsystems, it asks whether there exists a global density matrix $\rho$ such that all the given $\rho_{\I}$ can be obtained from $\rho$ as marginals. In this context, entropy inequalities provide necessary conditions for the existence of a solution \cite{Carlen_2013}.

The quantum marginal \textit{independence} problem (QMIP) \cite{Hernandez-Cuenca:2019jpv} can be interpreted as a simplified version of the quantum marginal problem, where instead of fixing the marginals for some subsystems, one only demands that certain subsystems are correlated and others are not. Specifically, the QMIP asks the following question: Given an $\N$-party system and a complete specification of the presence of correlation (or conversely, the lack thereof) among the various subsystems, is there a density matrix that satisfies these constraints? As for the case of the marginal problem, entropy inequalities can be used to constrain the space of solutions for the QMIP. However, unlike its parent version, the knowledge of the quantum entropy cone is sufficient to solve the problem completely. To see this, let us formulate the QMIP more precisely.

The $\N$-party \textit{subadditivity cone} (SAC) is defined as the polyhedral cone in entropy space carved out by all instances of SA for that $\N$. Consider a face $\face$ of the SAC and a vector $\vec\ent\in\text{int}(\face)$. Notice that the collection of SA instances which are saturated by $\vec\ent$ is independent from the specific choice of this vector. The saturation of SA is equivalent to the vanishing of the mutual information, which occurs if and only two subsystems are independent. We can then interpret $\text{int}(\face)$ to correspond to a specification of which subsystems are correlated and which are independent, while remaining agnostic about the specific values of the entropies. In its original formulation \cite{Hernandez-Cuenca:2019jpv}, the QMIP asked for which faces does $\text{int}(\face)$ contain at least one entropy vector that can be realized by a quantum state. We now slightly generalize this question by replacing ``realized'' with ``approximated arbitrarily well.'' In this generalized version, the QMIP can then be trivially solved if $\qec{\N}$ is known. 

In this paper we proceed in the opposite direction and attempt to extract information about $\qec{\N}$ from the solution to the QMIP for different values of $\N'\geq\N$. In particular, we will be interested in the extremal version of the QMIP, which we denote as EQMIP, where we only focus on the one-dimensional faces of the SAC, i.e., its extreme rays. Denoting by $\widehat{\er}_{\N}$ the set of all extreme rays of the $\N$-party SAC, the solution $\erq{\N}$ to the EQMIP is then
\begin{equation}
    \erq{\N}:=\widehat{\er}_{\N}\cap\qec{\N} .
\end{equation}
We will now explain how to construct inner bounds for $\qec{\N}$ from the knowledge of $\erq{\N'}$ for some $\N'\geq \N$.

%~~~~~~~~~~~~~~~~~~~~~~~~~~~~~~~~~~~~~~~~~~~~~~~~~~~~~~~~~~~~~~
\section{Inner bounds from extremal marginal independence} \label{sec:inner_bounds} 
%~~~~~~~~~~~~~~~~~~~~~~~~~~~~~~~~~~~~~~~~~~~~~~~~~~~~~~~~~~~~~~

%~~~~~~~~~~~~~~~~~~~~~~~~~~~~~~~~~~~~~~~~~~~~~~~~~~~~~~~~~~~~~~
\subsection{\label{subsec:cg} Subsystem coarse grainings}
%~~~~~~~~~~~~~~~~~~~~~~~~~~~~~~~~~~~~~~~~~~~~~~~~~~~~~~~~~~~~~~

Given a density matrix $\rho_{\N'}$, we want to consider a coarse graining of the $\N'$ parties into $\N$ composite ones and to relate the entropy vectors before and after the coarse graining. We will also consider purifications of $\rho_{\N'}$, and allow for coarse grainings of the ``purifying'' party. Defining $\llbracket\N\rrbracket=\{0,1,\ldots,\N\}$, where we added zero to account for the purifier, such a coarse graining can then be specified by a surjective function 
\begin{align}
\label{eq:f}
\begin{split}
    f:\; & \llbracket\N'\rrbracket\rightarrow\; \llbracket\N\rrbracket\\
    & \ell'\; \mapsto\; \ell=f(\ell'),
\end{split}
\end{align}
which specifies, for each party $\ell'\in\llbracket\N'\rrbracket$, which of the coarse gained $\llbracket\N\rrbracket$ parties it belongs to.

Given a density matrix $\rho_{\N'}$ with entropy vector $\vec{\ent}'$, and a coarse graining $f$, the components of the entropy vector of the \textit{same density matrix} after the coarse graining are then given by
\begin{equation}\label{eq:ent-coarsegrain}
    \ent'_{\I}=\ent_{f^{-1}(\I)} ,
\end{equation}
where $f^{-1}(\I)$ is the pre-image of $\I$ \footnote{If $0\in f^{-1}(\I)$, then $f^{-1}(\I)\not\subseteq [\N']$ and $\ent_{f^{-1}(\I)}$ is not a component of $\vec{\ent}(\rho_{\N'})$. However, the entropy of its complement in $\llbracket\N'\rrbracket$ is and is equal to  $\ent_{f^{-1}(\I)}$. This is what the right-hand side of Eq.~\eqref{eq:ent-coarsegrain} refers to implicitly.}.

Notice that in general, this map of entropy vectors from $\Sigma_{\N'}^*$ to $\Sigma_{\N}^*$, which we denote by $\Psi_f$, can formally be extended to vectors in the full $\qec{\N'}$, and can map a vector on the boundary of $\qec{\N'}$ to one strictly in the interior of $\qec{\N}$. For example, any coarse graining to $\N=2$ of the $\N=3$ entropy vector corresponding to the four-party Greenberger-Horne-Zeilinger (GHZ) state gives an entropy vector strictly inside $\qec{2}$ (since no instance of SA is saturated).

%~~~~~~~~~~~~~~~~~~~~~~~~~~~~~~~~~~~~~~~~~~~~~~~~~~~~~~~~~~~~~~
\subsection{Inner bounds for $\qec{\N}$}
\label{subsec:bounds}
%~~~~~~~~~~~~~~~~~~~~~~~~~~~~~~~~~~~~~~~~~~~~~~~~~~~~~~~~~~~~~~

Using this construction, we can now derive an inner bound for $\qec{\N}$ as follows. Suppose that $\erq{\N'}$ is known for some $\N'\geq \N$. We can then consider all possible coarse graining from $\N'$ to $\N$, mapping all vectors in $\erq{\N'}$ to $\N$-party entropy vectors, and taking the conical hull. Formally, for any given $\N$ and $\N'\geq\N$, we define
\begin{equation}
\label{eq:delta}
    \Delta_{\N}^{\N'}=\text{cone}\,\big\{\Psi_f(\vec{{\sf R}}),\;\forall f\in\Phi\;\;\forall\, \vec{\sf R}\in \erq{\N'} \big\} ,
\end{equation}
where $\Phi$ is the set of all possible functions $f$ from Eq.~\eqref{eq:f}. We then have the following lemma.
\begin{lemma}
\label{lemma:main}
    For any fixed $\N$, and any $i,j\in\mathbb{N}$ with $i<j$, $\Delta_{\N}^{\N+i}\subseteq\Delta_{\N}^{\N+j}$.
\end{lemma}
\begin{proof}
    It suffices to show that for any $\N'\geq \N$ and $\vec{\sf R}\in \erq{\N}$, there exists some $\vec{\sf R}'\in \erq{\N'}$ and a coarse graining $f$ from $\N'$ to $\N$ such that $\Psi_f(\vec{\sf R}')=\vec{\sf R}$. For the case where $\vec{\sf R}$ is realized by a density matrix $\rho$, we can realize the desired $\vec{\sf R}'$ with $\rho'=\rho\otimes\ket{{\bf 0}}\!\bra{{\bf 0}}_{[\N']\setminus[\N]}$, and choose $f$ such that the additional factors are coarse gained with the purifier for $\rho'$. For the general case, we refer the reader to Section 3.4 of \cite{He:2022bmi}. 
\end{proof}

By \Cref{lemma:main} and the definition of $\Delta_{\N}^{\N'}$ in Eq.~\eqref{eq:delta}, we obtain the following chain of inclusions:
\begin{equation}
\label{eq:chain}
    \Delta_{\N}^{\N}\subseteq \Delta_{\N}^{\N+1}\subseteq\Delta_{\N}^{\N+2}\subseteq\ldots\subseteq\qec{\N}\subseteq\poly{\N} .
\end{equation}
It is then natural to ask how well this sequence approximates $\qec{\N}$, and if there is some ``maximal'' $\widehat{i}$ (which implicitly depends on $\N$) such that
\begin{equation}
\label{eq:stop}
    \Delta_{\N}^{\N+\widehat{i}}=\Delta_{\N}^{\N+i},\quad \forall\, i>\widehat{i}.
\end{equation}

In the following sections, we will answer this question for $\N=3$ and provide a partial answer for $\N=4$. First, however, we need to review a construction that will allow us to determine at least a subset of solutions to the EQMIP for sufficiently large $\N$.

%~~~~~~~~~~~~~~~~~~~~~~~~~~~~~~~~~~~~~~~~~~~~~~~~~~~~~~~~~~~~~~
\section{\label{sec:hypergraphs} Hypergraph models}
%~~~~~~~~~~~~~~~~~~~~~~~~~~~~~~~~~~~~~~~~~~~~~~~~~~~~~~~~~~~~~~

The \textit{hypergraph models} of entanglement were introduced in \cite{Bao:2020zgx} as a generalization of the graph models used in \cite{Bao:2015bfa} to study entropy inequalities in quantum gravity. An $\N$-party hypergraph model is a simple \footnote{By simple we mean the hypergraph has no loops and no repeated edges.} weighted hypergraph $H=(V,E)$ with vertices $V$, hyperedges $E$ with positive weights, a specification of a subset $\partial V\subseteq V$ of vertices called \emph{boundary vertices}, and a surjective (but not necessarily injective) map $\xi:\partial V\rightarrow \llbracket\N\rrbracket$. 

Given an $\N$-party hypergraph model, one associates to it an entropy vector as follows. For a non-empty subset $\I\subseteq [\N]$, an $\I$-cut is a subset $V_{\I}\subset V$ such that $\partial V\cap V_{\I}=\xi^{-1}(\I)$, where $\xi^{-1}$ denotes the pre-image. The \textit{cost} of any such cut is the sum of the weights of the hyperedges that connect a vertex in $V_{\I}$ to one in $V_{\I}^c$, the complement of $V_{\I}$ in $V$. In other words, given an $\I$-cut $V_{\I}$, and a hyperedge $h$ (thought of as a collection of vertices), the weight of $h$ contributes to the cost of the cut if and only if $h$ contains at least one vertex in $V_{\I}$ and at least one in $V_{\I}^c$. The entropy $\ent_{\I}$ is then defined as the cost of the $\I$-cut with \textit{minimal cost}.

The prescription we just described associates to each $\N$-party hypergraph model a vector in $\mathbb{R}^{\D}$. The collection of all these vectors (at fixed $\N$) is again a convex cone \footnote{A conical combination of two entropy vectors realized by hypergraphs is realized by the hypergraph obtained by the disjoint union of the hypergraphs with the weights appropriately rescaled.}, and we have the following theorem \cite{Walter:2020zvt}.

\begin{thm}
\label{thm:main}
    The entropy cone of hypergraph models is contained in the entropy cone of stabilizer states. 
\end{thm}

While in principle hypergraph models may not be sufficient to solve the EQMIP completely at arbitrary $\N$ (see \S\ref{sec:discussion}), they at least provide a partial solution that will be sufficient for our purposes.

\begin{figure}[tb]
    \centering
    \begin{tikzpicture}
    \node[] at (0,0) {\includegraphics{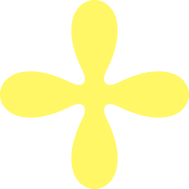}};
    \draw[very thick,blue]  (0,1) -- (-1,2) ;
    \draw[very thick,blue]  (0,1) -- (1,2);
    \draw[very thick,blue]  (1,0) -- (2,1);
    \draw[very thick,blue]  (1,0) -- (2,-1);
    \draw[very thick,blue]  (0,-1) -- (1,-2);
    \draw[very thick,blue]  (0,-1) -- (-1,-2);
    \filldraw[red] (0,1) circle (2.5pt);
    \filldraw[red] (1,0) circle (2.5pt);
    \filldraw[red] (0,-1) circle (2.5pt);
    \filldraw (-1,2) circle (2.5pt); %B
    \filldraw (1,2) circle (2.5pt); %C
    \filldraw (-1,0) circle (2.5pt); %A
    \filldraw (2,1) circle (2.5pt); %D
    \filldraw (2,-1) circle (2.5pt); %E
    \filldraw (1,-2) circle (2.5pt); %F
    \filldraw (-1,-2) circle (2.5pt); %O
    \draw[orange] (0,0) node[]{{\footnotesize $2$}};
    \draw[] (-1.4,0) node{{\footnotesize $0$}};
    \draw[] (-1.3,2.3) node{{\footnotesize $1$}};
    \draw[] (1.3,2.3) node{{\footnotesize $2$}};
    \draw[] (2.3,1.3) node{{\footnotesize $3$}};
    \draw[] (2.3,-1.3) node{{\footnotesize $4$}};
    \draw[] (1.3,-2.3) node{{\footnotesize $5$}};
    \draw[] (-1.3,-2.3) node{{\footnotesize $6$}};
    \end{tikzpicture}
    \caption{A hypergraph giving an element of $\erq{6}$. The hyperedge (yellow blobs) has weight 2, and edges (blue lines) have weight 1.}
    \label{fig:hg1}
\end{figure}
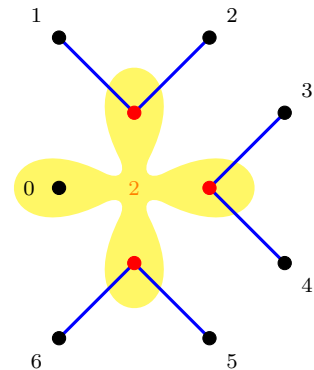

\begin{figure}[tb]
    \centering
    \begin{tikzpicture}
    \node[] at (0,0) {\includegraphics{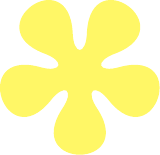}};
    \draw[very thick,blue]  (270:1) -- (290:2);
    \draw[very thick,blue]  (270:1) -- (250:2);
    \draw[very thick,blue]  (198:1) -- (218:2);
    \draw[very thick,blue]  (198:1) -- (178:2);
    \draw[very thick,blue]  (126:1) -- (146:2);
    \draw[very thick,blue]  (126:1) -- (106:2);
    \draw[very thick,blue]  (54:1) -- (74:2);
    \draw[very thick,blue]  (54:1) -- (34:2);
    \draw[very thick,blue]  (-18:1) -- (2:2);
    \draw[very thick,blue]  (-18:1) -- (-38:2);
    \filldraw[red] (270:1) circle (2.5pt);
    \filldraw[red] (198:1) circle (2.5pt);
    \filldraw[red] (126:1) circle (2.5pt);
    \filldraw[red] (54:1) circle (2.5pt);
    \filldraw[red] (-18:1) circle (2.5pt);
    \filldraw (250:2) circle (2.5pt); 
    \filldraw (290:2) circle (2.5pt); 
    \filldraw (218:2) circle (2.5pt);
    \filldraw (178:2) circle (2.5pt);
    \filldraw (146:2) circle (2.5pt);
    \filldraw (106:2) circle (2.5pt);
    \filldraw (74:2) circle (2.5pt);
    \filldraw (34:2) circle (2.5pt);
    \filldraw (2:2) circle (2.5pt);
    \filldraw (-38:2) circle (2.5pt);
    \draw[orange] (0,0) node[]{{\footnotesize $2$}};
    \draw[] (290:2.3) node{{\footnotesize $1$}};
    \draw[] (250:2.3) node{{\footnotesize $2$}};
    \draw[] (218:2.3) node{{\footnotesize $3$}};
    \draw[] (178:2.3) node{{\footnotesize $4$}};
    \draw[] (146:2.3) node{{\footnotesize $5$}};
    \draw[] (106:2.3) node{{\footnotesize $6$}};
    \draw[] (74:2.3) node{{\footnotesize $7$}};
    \draw[] (34:2.3) node{{\footnotesize $8$}};
    \draw[] (2:2.3) node{{\footnotesize $9$}};
    \draw[] (-38:2.3) node{{\footnotesize $0$}};
    \end{tikzpicture}
    \caption{A hypergraph giving an element of $\erq{9}$. The hyperedge (yellow blobs) has weight 2, and edges (blue lines) have weight 1.}
    \label{fig:hg2}
\end{figure}

\begin{figure}[tb]
    \centering
    \begin{tikzpicture}
    \node[] at (0,0) {\includegraphics{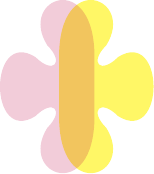}};
    \draw[very thick,blue]  (0,1) -- (-1,2) ;
    \draw[very thick,blue]  (0,1) -- (1,2);
    \draw[very thick,blue]  (1,0) -- (2,1);
    \draw[very thick,blue]  (1,0) -- (2,-1);
    \draw[very thick,blue]  (-1,0) -- (-2,1);
    \draw[very thick,blue]  (-1,0) -- (-2,-1);
    \draw[very thick,blue]  (0,-1) -- (0,-2.414);
    \draw[very thick,blue]  (0,0) -- (-1.5,-1.5);
    \draw[very thick,blue]  (0,0) -- (1.5,-1.5);
    \filldraw[red] (0,1) circle (2.5pt);
    \filldraw[red] (1,0) circle (2.5pt);
    \filldraw[red] (0,-1) circle (2.5pt);
    \filldraw[red] (-1,0) circle (2.5pt); 
    \filldraw[red] (0,0) circle (2.5pt); 
    \filldraw (-1,2) circle (2.5pt); 
    \filldraw (1,2) circle (2.5pt);  
    \filldraw (2,1) circle (2.5pt); 
    \filldraw (2,-1) circle (2.5pt); 
    \filldraw (-2,1) circle (2.5pt); 
    \filldraw (-2,-1) circle (2.5pt); 
    \filldraw (-1.5,-1.5) circle (2.5pt);
    \filldraw (1.5,-1.5) circle (2.5pt);
    \filldraw (0,-2.414) circle (2.5pt);  
    \draw[purple] (-0.6,0) node[]{{\footnotesize $2$}};
    \draw[orange] (0.6,0) node[]{{\footnotesize $2$}};
    \draw[blue] (0.01,-1.9) node[right]{{\footnotesize $2$}};
    \draw[blue] (-1.1,1.6) node[right]{{\footnotesize $2$}};
    \draw[blue] (0.7,1.6) node[right]{{\footnotesize $2$}};
    \draw[blue] (-1.5,-1) node[right]{{\footnotesize $2$}};
    \draw[blue] (1.1,-1) node[right]{{\footnotesize $2$}};
    \draw[] (-1.3,2.3) node{{\footnotesize $3$}};
    \draw[] (1.3,2.3) node{{\footnotesize $4$}};
    \draw[] (2.3,1.3) node{{\footnotesize $5$}};
    \draw[] (2.3,-1.3) node{{\footnotesize $6$}};
    \draw[] (-2.3,1.3) node{{\footnotesize $2$}};
    \draw[] (-2.3,-1.3) node{{\footnotesize $1$}};
    \draw[] (0,-2.714) node{{\footnotesize $0$}};
    \draw[] (-1.8,-1.8) node{{\footnotesize $7$}};
    \draw[] (1.8,-1.8) node{{\footnotesize $8$}};
    \end{tikzpicture}
    \caption{A hypergraph giving an element of $\erq{8}$. The two hyperedges (yellow and purple blobs) have weight 2, and edges (blue lines) have weights either 1 (when unspecified) or 2.}
    \label{fig:hg3}
\end{figure}
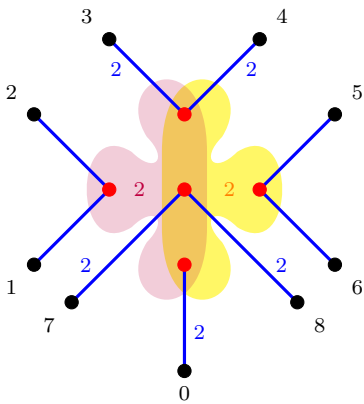

%~~~~~~~~~~~~~~~~~~~~~~~~~~~~~~~~~~~~~~~~~~~~~~~~~~~~~~~~~~~~~~
\section{Constructions for small $\N$}
\label{sec:examples}
%~~~~~~~~~~~~~~~~~~~~~~~~~~~~~~~~~~~~~~~~~~~~~~~~~~~~~~~~~~~~~~

Having reviewed the necessary tools from hypergraph models, we now use this technology to construct inner bounds for $\qec{\N}$ for $\N\leq 4$. The logic will be as follows: we will construct hypergraphs that realize certain extreme rays of the SAC which, by \Cref{thm:main}, are automatically solutions to the EQMIP. Even without a full solution to the EQMIP, we will then show via particular coarse grainings that we can obtain entropy vectors for a smaller number of parties, which give powerful inner bounds.

%~~~~~~~~~~~~~~~~~~~~~~~~~~~~~~~~~~~~~~~~~~~~~~~~~~~~~~~~~~~~~~
\subsection{\label{subsec:n2} Bounds for $\N=2$}
%~~~~~~~~~~~~~~~~~~~~~~~~~~~~~~~~~~~~~~~~~~~~~~~~~~~~~~~~~~~~~~

This case is trivial, but we include it for completeness. The outer bound $\poly{2}$ to $\qec{2}$ is the SAC, and its extreme rays are realized by Bell pairs. We then have 
\begin{equation}
    \Delta_2^2=\qec{2}=\poly{2} .
\end{equation}
%

%~~~~~~~~~~~~~~~~~~~~~~~~~~~~~~~~~~~~~~~~~~~~~~~~~~~~~~~~~~~~~~
\subsection{\label{subsec:n3} Bounds for $\N=3$}
%~~~~~~~~~~~~~~~~~~~~~~~~~~~~~~~~~~~~~~~~~~~~~~~~~~~~~~~~~~~~~~

As we mentioned in \S\ref{sec:cones}, it was already shown in \cite{1193790} that $\qec{3}=\poly{3}$. This equivalence follows from the fact that the extreme rays of $\poly{3}$ contain entropy vectors realized by Bell pairs, the four-party GHZ state, and the four-party absolutely maximally entangled state (also known as the four-party perfect tensor) \cite{Helwig:2013ckb}. With the only exception of the four-party GHZ, all these states also generate extreme rays of the three-party SAC and are hence elements of $\erq{3}$. Their conical hull is $\Delta_3^3$, which coincides with the three-party holographic entropy cone \cite{Bao:2015bfa}. 

To show that the sequence in Eq.~\eqref{eq:chain} converges to $\qec{3}$, we need to consider $\Delta_3^{3+i}$ for $i\geq 1$. The sets $\erq{4}$ and $\erq{5}$ were already found in \cite{Hernandez-Cuenca:2019jpv}, and it is straightforward to verify that they imply $\Delta_3^3=\Delta_3^4=\Delta_3^5$. Hence, we need $i\geq 3$. It turns out that $i=3$ is sufficient, since one can verify that the entropy vector associated to the hypergraph in \Cref{fig:hg1} is in $\erq{6}$, and under the coarse graining
\begin{equation}
\label{eq:N6toN3cg}
    \{0,1,2,3,4,5,6\} \rightarrow \{0,1,1,2,2,3,3\} ,
\end{equation}
results in a hypergraph model of the GHZ state. In summary, we have
\begin{equation}
\label{eq:sequence3}
    \Delta_3^3=\Delta_3^4=\Delta_3^5\subset\Delta_3^6=\qec{3},
\end{equation}
which by Eq.~\eqref{eq:chain} also implies $\Delta_3^{6+i}=\Delta_3^6$ for all $i\in\mathbb{N}$.

%~~~~~~~~~~~~~~~~~~~~~~~~~~~~~~~~~~~~~~~~~~~~~~~~~~~~~~~~~~~~~~
\subsection{\label{subsec:n4} Bounds for $\N=4$}
%~~~~~~~~~~~~~~~~~~~~~~~~~~~~~~~~~~~~~~~~~~~~~~~~~~~~~~~~~~~~~~

As mentioned in \S\ref{sec:cones}, the QEC for $\N=4$ is not known. However, we will show that our construction, even with a relatively small value of $\N'$, suffices to derive an inner bound that coincides with the subset of $\qec{4}$ of entropy vectors realized by stabilizer states. This subset, which we denote by $\overline{\Lambda}_{4}^*$, was shown in \cite{Linden:2013kal} to coincide with its outer bound $\Lambda_{4}$, defined as the polyhedral cone specified by SA, SSA, and Ingleton's inequality \cite{zbMATH03351601}. Our goal is to show that there is some $i\in\mathbb{N}$ such that 
\begin{equation}
    \Delta_4^{4+i}=\overline{\Lambda}_{4}^*.
\end{equation}

To see this, we can use a result from \cite{Bao:2020zgx}, which constructed hypergraph models realizing the entropy vectors that generate the extreme rays of $\Lambda_4$. It is then enough to show that each of these models (seven in total) can be obtained 
from coarse graining another hypergraph involving $\N'$ parties (for some $\N'\geq 4$), whose corresponding entropy vector is an extreme ray of the $\N'$-party SAC. Four out of these seven cases are trivial, since they correspond to elements of $\erq{4}$. For the remaining three non-trivial cases (which require $i>0$), we construct the hypergraphs shown in Figures~\ref{fig:hg1}--\ref{fig:hg3}. As one can verify, they are all associated to extreme rays of the SAC for their corresponding number of parties. By an appropriate choice of coarse grainings, we can reduce these hypergraphs to the ones labeled by 4, 6 and 7 in Figure~4 of \cite{Bao:2020zgx} (up to a trivial permutation of the parties). Specifically, these coarse grainings are
\begin{align}
\begin{split}
      \{0,1,2,3,4,5,6\} &\rightarrow \{0,1,1,2,2,3,4\}     \nonumber\\
      \{0,1,2,3,4,5,6,7,8,9\} &\rightarrow \{0,1,1,2,2,3,3,4,4,0\}      \nonumber\\
      \{0,1,2,3,4,5,6,7,8\} &\rightarrow \{0,1,1,2,2,3,3,4,4\}  .
\end{split}
\end{align} 
We then have
\begin{equation}
\label{eq:chain4}                    \Delta_4^4=\Delta_4^5\subset\Delta_4^6\subset\overline{\Lambda}_{4}^*\subseteq\Delta_4^9\subseteq\qec{4}\subseteq\poly{4},
\end{equation}
where the first equality follows trivially from the explicit knowledge of $\erq{4}$ and $\erq{5}$ \cite{Hernandez-Cuenca:2019jpv}, and the strict inclusion of $\Delta_4^6$ in $\overline{\Lambda}_{4}^*$ follows from the results of \cite{He:2022wip} regarding $\erq{6}$. It is possible that $\erq{8}$, or even $\erq{7}$, might be enough to obtain $\overline{\Lambda}_{4}^*$, and we leave this for future work.

%~~~~~~~~~~~~~~~~~~~~~~~~~~~~~~~~~~~~~~~~~~~~~~~~~~~~~~~~~~~~~~
\section{\label{sec:discussion} Discussion}
%~~~~~~~~~~~~~~~~~~~~~~~~~~~~~~~~~~~~~~~~~~~~~~~~~~~~~~~~~~~~~~

We conclude by commenting on a few open questions regarding the strength of our bound, which we believe to be particularly interesting. 

For the bounds that we have explicitly constructed in this paper, we relied on the hypergraph models of \cite{Bao:2020zgx} and the result of \cite{Walter:2020zvt}. However, it was shown in \cite{Bao:2020mqq} that for $\N\geq 5$ the inclusion of the entropy cone of hypergraphs in the cone of stabilizer states is strict. This implies that hypergraph models might not be sufficient to derive the most stringent possible bounds. On the other hand, we cannot rule out the possibility that hypergraph models do suffice to solve the EQMIP for arbitrary $\N$, in which case our bounds would not converge to the full QEC for any $\N\geq 4$. To resolve this, one direction is to look for elements in  $\erq{\N}$ that can be realized by stabilizer states but violate the ``hypergraph'' inequality found in \cite{Bao:2020mqq} at $\N=5$.  

A similar limitation might also affect stabilizer states, since it is possible that for some value of $\N$, the solution to the EQMIP requires non-stabilizer quantum states. This question is related to the most stringent bound that can be realized at $\N=4$. Indeed, any $\Delta_4^{4+i}$ such that $\overline{\Lambda}_{4}^*\subset\Delta_4^{4+i}$ requires an element of $\erq{4+i}$ that cannot be realized by a stabilizer state, since it has to violate Ingleton's inequality. Such violations have already been investigated for classical probability distributions \cite{Mao2009ViolatingTI}, and another question for the future is to determine if for some $\N$ there exist extreme rays of the SAC that violate Ingleton's inequality while being realizable by quantum states (at least approximately). 

More generally, deep structural properties of the QEC might be extractable by exploring the behavior of the sequence of approximations given in Eq.~\eqref{eq:chain}. 
%\MR{[new version]} 
As we already mentioned at the end of \S\ref{sec:inner_bounds}, one possibility (for any given $\N$) is that there exists some $\widehat{i}$ such that the sequence does not change beyond this point (see Eq.~\eqref{eq:stop}). This is the case for $\N=3$, where we have shown that the approximation is exact, and the sequence already converges to the full $\qec{3}$ at $i=3$ (see Eq.~\eqref{eq:sequence3}).  However, in principle this behavior is possible even if the approximation is not exact. For example, this would be the case for $\N=4$ if stabilizer states were sufficient to solve the EQMIP. Another possibility is that (at least for some $\N$) the sequence continues to grow indefinitely and no such $\widehat{i}$ exists. In that case, it would be important to understand if it converges, and how its limit is related to $\qec{\N}$.

These possibilities are also related to the question about the polyhedrality of the QEC. For $\N=3$, the exact approximation of $\qec{3}$ at finite $\widehat{i}$ is possible only because $\qec{3}$ is a polyhedral cone, since any $\Delta_{\N}^{\N'}$ is polyhedral by construction. For $\N>3$, it is currently not known whether $\qec{\N}$ is polyhedral. If it is not, as is the case for the subset of $\qec{\N}$ corresponding to classical probability distributions \cite{2007:Matus}, the sequence in Eq.~\eqref{eq:chain} might still converge, but only at infinite $i$.

Finally, let us comment on the role of the outer bound $\poly{\N}$, and in particular SSA. As we mentioned in the introduction, to our knowledge there is currently no strong evidence against the possibility that the QEC also coincides with $\poly{\N}$ for $\N>3$. To explore this possibility, a natural way to proceed is to examine the gap between our inner bound and $\poly{\N}$. In general, an extreme ray of $\poly{\N}$ is the intersection of facets corresponding to multiple instances of both SA and SSA, and it may not be an extreme ray of the $\N$-party SAC. Thus, one may wonder if such an extreme ray can be obtained from coarse graining an extreme ray of the SAC involving more parties and realizable (at least approximately) in quantum mechanics. It is intriguing that an example of this is already present at $\N=3$. In this case, one of the extreme rays of $\poly{3}$, namely the one corresponding to the four-party GHZ state described in \S\ref{subsec:n3}, does not saturate any instance of SA, and nevertheless, it can be obtained from our construction via the coarse graining in Eq.~\eqref{eq:N6toN3cg} of the element of $\erq{6}$ shown in \Cref{fig:hg1}. This suggests that further explorations of the gap between our inner bound and the outer bound $\poly{\N}$ might provide powerful insights on the structure of the quantum entropy cone.

\vspace{0.5cm}

\noindent {\it Acknowledgments:} T.H. is supported by the Heising-Simons Foundation ``Observational Signatures of Quantum Gravity'' collaboration grant 2021-2817, the U.S. Department of Energy grant DE-SC0011632, and the Walter Burke Institute for Theoretical Physics. V.H. has been supported in part by the U.S. Department of Energy grant DE-SC0009999 and by funds from the University of California. M.R is supported by funds from the University of California.

\bibliography{entropy-cone-bib} 

\end{document}